\documentclass[journal]{IEEEtran}

\usepackage[utf8]{inputenc}
\usepackage[T1]{fontenc}
\usepackage{url}
\usepackage{ifthen}
\usepackage{cite}
\usepackage[cmex10]{amsmath} 

\usepackage{mypackage}
\usepackage{mathrsfs}

\usepackage{mathtools}

\DeclarePairedDelimiter\floor{\lfloor}{\rfloor}
\usepackage{float}

\newtheorem{theorem}{Theorem}

\usepackage{physics}
\usepackage{siunitx}

\newcommand{\sir}{\mathrm{SIR}}
\newcommand{\Pb}{\mathbb{P}}
\newcommand{\Eb}{\mathbb{E}}

\newcommand{\Lc}{\mathcal{L}}

\newcommand{\black}{\textcolor{black}}

\usepackage{multicol}

\usepackage[nomain,acronym,shortcuts]{glossaries}
\makeglossaries
\newcommand*{\acro}[3][]{\newacronym[#1]{#2}{#2}{#3}}

\acro{D2D}{device-to-device}
\acro{SIR}{signal-to-interference-ratio}
\acro{SINR}{signal-to-interference-plus-noise-ratio}
\acro{PCP}{Poisson cluster process}
\acro{CoMP}{coordinated multi-point}
\acro{BS}{base station} 
\acro{MD-CoMP}{macrodiversity CoMP transmission}
\acro{MAC}{medium-access-control}
\acro{JT-CoMP}{joint transmission CoMP}
\acro{CoMP-JT}{coordinated multipoint joint transmission}
\acro{SBS}{small base station}
\acro{MDSD}{multiple devices to the single device}
\acro{MDS}{maximum distance seperable}
\acro{SCN}{small cell network}
\acro{PPP}{Poisson point process}
\acro{TCP}{Thomas cluster process}
\acro{CSI}{channel state information}
\acro{PDF}{probability distribution function}
\acro{PMF}{probability mass function}
\acro{RV}{random variable}
\acro{i.i.d.}{independently and identically distributed}
\acro{MBMS}{multimedia broadcasting multicasting service}
\acro{EE}{energy efficiency}
\acro{HCP}{hard-core placement}
\acro{CCDF}{complementary cumulative distribution function}
\acro{CDF}{cumulative distribution function}
\acro{PC}{probabilistic caching}
\acro{RC}{random uniform caching}
\acro{CPF}{caching popular files} 
\acro{PGFL}{probability generating functional}
\acro{KKT}{Karush-Kuhn-Tucker}
\acro{PGF}{point generating function}
\acro{BPP}{binomial point process}
\acro{3GPP}{3rd Generation Partnership Project}
\acro{UAV}{unmanned aerial vehicle}
\acro{FAA}{Federal Aviation Administration}
\acro{UAS}{unmanned aircraft systems}
\acro{AP}{access point}
\acro{GSM}{Global System for Mobile Communications}
\acro{LTE}{long term evolution}
\acro{IoT}{Internet-of-things}
\acro{UE}{user equipment}
\acro{CNPC}{control and non-payload communication}
\acro{FHD}{full high definition}
\acro{NLoS}{non-line-of-sight}
\acro{LoS}{line-of-sight}
\acro{5G}{fifth-generation}
\acro{MGF}{moment generating function}

\usepackage{amssymb}

\begin{document}
\title{Caching to the Sky: Performance Analysis of Cache-Assisted CoMP for Cellular-Connected UAVs}

\author{\IEEEauthorblockN{Ramy Amer\IEEEauthorrefmark{1},			
Walid Saad\IEEEauthorrefmark{2},
Hesham ElSawy\IEEEauthorrefmark{3},
M. Majid Butt\IEEEauthorrefmark{1}\IEEEauthorrefmark{4},
Nicola Marchetti\IEEEauthorrefmark{1}}\\
\IEEEauthorblockA{\IEEEauthorrefmark{1}CONNECT, Trinity College, University of Dublin, Ireland}\\
\IEEEauthorblockA{\IEEEauthorrefmark{2}Wireless@VT, Bradley Department of Electrical and Computer Engineering, Virginia Tech, Blacksburg, VA, USA}\\
\IEEEauthorblockA{\IEEEauthorrefmark{3}King Fahd University of Petroleum and Minerals (KFUPM)}\\
\IEEEauthorblockA{\IEEEauthorrefmark{4}Nokia Bell Labs, France}\\
\IEEEauthorblockA{email:\{ramyr, majid.butt, nicola.marchetti\}@tcd.ie, walids@vt.edu, hesham.elsawy@kfupm.edu.sa}}

\maketitle
\begin{abstract}
Providing connectivity to aerial users, such as cellular-connected unmanned aerial vehicles (UAVs) or flying taxis, is a key challenge for tomorrow's cellular systems. In this paper, the use of coordinated multi-point (CoMP) transmission along with caching for providing seamless connectivity to aerial users is investigated. In particular, a network of clustered cache-enabled small base stations (SBSs) serving aerial users is considered in which a requested content by  an aerial user is cooperatively transmitted from collaborative ground SBSs. For this network, a novel upper bound expression on the coverage probability is derived as a function of the system parameters. The effects of various system parameters such as collaboration distance and content availability on the achievable performance are then investigated. Results reveal that, when the antennas of the SBSs are tilted downwards, the coverage probability of a high-altitude aerial user is upper bounded by that of a ground user regardless of the transmission scheme. Moreover, it is shown that for a low signal-to-interference-ratio (SIR) threshold, CoMP transmission improves the coverage probability for aerial users from $10\%$ to $70\%$ under a collaboration distance of \SI{200}{m}.
 
\end{abstract}

\section{Introduction}	
The past few years have witnessed a tremendous increase in the use of \acp{UAV}, also known as drones, in many applications, such as aerial surveillance, package delivery, and even flying taxis \cite{8473483}. Enabling such \ac{UAV}-centric applications requires ubiquitous wireless connectivity that can be potentially provided by the popular wireless cellular network \cite{challita2018deep,mozaffari2018tutorial,mozaffari2016unmanned}. However, in order to operate such cellular-connected \acp{UAV} using existing wireless systems, one must address a broad range of challenges that include resource management, interference mitigation, power control, and energy efficiency \cite{zeng2018cellular15}.	
 
In particular, the altitude of a cellular-connected \ac{UAV} will be much higher than the typical terrestrial \ac{UE} and will significantly exceed the \ac{SBS} antenna height. Consequently, ground \acp{SBS} need to be able to provide 3D communication coverage. However, existing \acp{SBS} antennas are usually tilted downwards to cater to the ground coverage and reduce inter-cell interference. Preliminary field measurement results by Qualcomm have demonstrated adequate aerial coverage by \ac{SBS} antenna sidelobes for \acp{UAV} below \SI{120}{m} \cite{qualcomm2017unmanned}. However, as the \ac{UAV} altitude further increases, new solutions are needed to enable cellular \acp{SBS} to seamlessly cover the sky \cite{azari2017coexistence,azari2018reshaping}. 

The dominance of \ac{LoS} links makes inter-cell interference a critical issue for cellular systems with hybrid  terrestrial and aerial \acp{UE}.
In this regard, extensive real-world simulations and fields trials in \cite{zeng2018cellular15,lin2018sky}, and \cite{van2016lte} have shown that an aerial \ac{UE}, in general, has poorer downlink performance than a ground \ac{UE}. Due to the down-tilted \ac{SBS} antennas, it is found that \acp{UAV} at \SI{40}{m} and higher are served by the sidelobes of \ac{SBS} antennas, which have reduced antenna gain compared to the mainlobes of \ac{SBS} antennas serving ground \acp{UE}. However, \acp{UE} at \SI{40}{m} and above experience free space propagation conditions, while radio signals attenuate more quickly with distance on the ground. Interestingly, it is shown that free space propagation can make up for the \ac{SBS} antenna sidelobe gain reduction \cite{lin2018sky}. However, this merit from such a favorable  \ac{LoS} channel vanishes at high altitudes because the free space propagation also leads to stronger \ac{LoS} interfering signals. Eventually, aerial \acp{UE} at high altitudes are shown to always have poorer communication and coverage as opposed to ground \acp{UE} \cite{lin2018sky,zeng2018cellular15}, and \cite{van2016lte}. 
\black{Nevertheless, while interesting, these works explored the feasibility of providing cellular connectivity for \acp{UAV} without proposing new approaches to solve the ensuing \ac{LoS}-dominated interference problem.}




In \cite{azari2018reshaping}, the authors studied the feasibility of supporting drone operations using existent cellular infrastructure. The authors showed that carefully designed system parameters such as antenna radiation pattern and network density guarantee a satisfactory quality of service.
Meanwhile, the authors in \cite{zhang2018cellular55} considered a cellular-connected \ac{UAV} flying from an initial location to a final destination. The authors minimized the UAV's mission completion time by optimizing its trajectory while maintaining reliable communication with the ground cellular network.
\black{However, while the works in \cite{azari2017coexistence,azari2018reshaping}, and \cite{zhang2018cellular55} have analyzed the performance of cellular-connected \acp{UAV}, effectively mitigating the impact of \ac{LoS}
interference on contemporary aerial  \acp{UE} has not yet been addressed in the literature.}



Compared with this prior art \cite{zeng2018cellular15,qualcomm2017unmanned,azari2017coexistence,azari2018reshaping,
lin2018sky,van2016lte,zhang2018cellular55}, the main contribution of this paper is to develop a novel framework that leverages \ac{CoMP} transmissions for serving high-altitude cellular-connected \acp{UAV} while mitigating cross-cell interference and boosting the received \ac{SIR}. In particular, we consider a network of cache-enabled \acp{SBS} in which an aerial \ac{UE} downloads a previously cached content via \ac{CoMP} transmission from neighboring ground \acp{SBS}. Using tools from stochastic geometry, we derive a considerably tight upper bound on the content coverage probability as a function of the system parameters. We show that the achievable performance of an aerial user depends heavily on the collaboration distance, content availability, and target bit rate. Moreover, while allowing \ac{CoMP} transmission substantially improves the coverage probability for aerial \acp{UE}, it is shown that their performance is still upper bounded by that of ground \acp{UE} due to the down-tilt of the current SBSs' antennas.
To the best of our knowledge, this paper provides the first analysis of \ac{CoMP} transmission for cellular-connected \acp{UAV} in cache-enabled networks.

The rest of this paper is organized as follows. Section II and Section III present the system model and the coverage probability analysis, respectively. Numerical results are presented in Section IV and conclusions are drawn in Section V.


\section{System Model}
 \begin{figure} [!t]	
\centering
\includegraphics[width=0.40\textwidth]{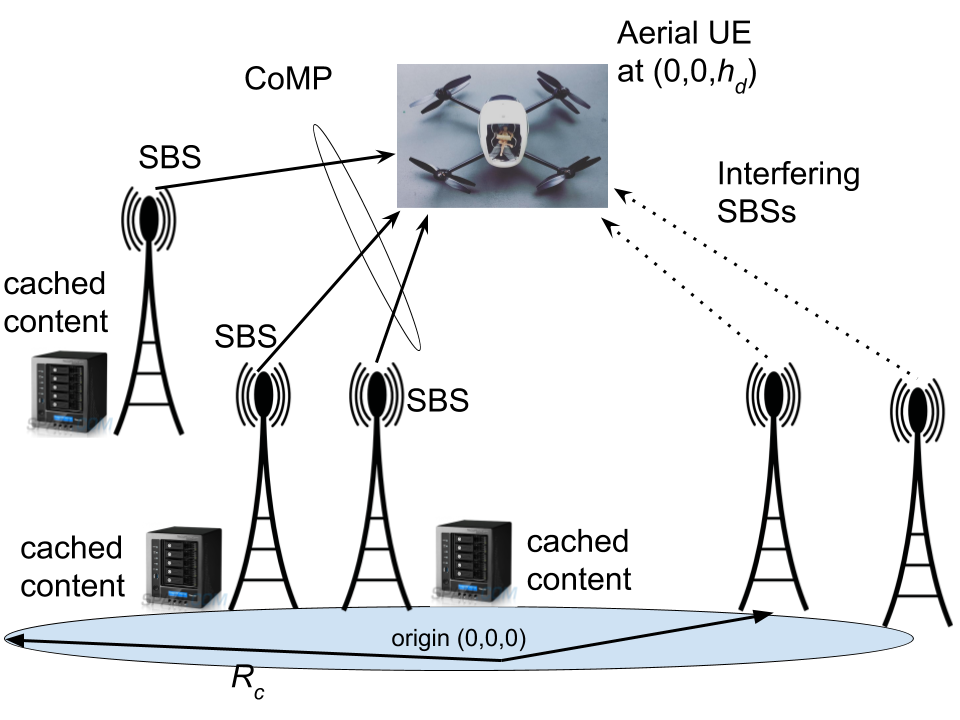}		
\caption {Illustration of the proposed system model.}
\label{system-model-comp}
\vspace{-0.4 cm}
\end{figure}

We consider a cache-enabled small cell network in which \acp{SBS} are distributed according to a homogeneous \ac{PPP} 
$\Phi_b=\{ b_i \in \mathbb{R}^2, \forall i \in \mathbb{N}^+ \}$ with intensity $\lambda_b$. We consider a cellular-connected aerial \ac{UE} flying at an altitude $h_d$ and located at $(0,0,h_d) \in \mathbb{R}^3$, where $h_d$ is the altitude of the \ac{UAV}. 
\black{We assume a user-centric model in which the \acp{SBS} are grouped into disjoint clusters around aerial \acp{UE} to be served \cite{8449213}. A cluster is represented by a circle of radius $R_c$ centered at the terrestrial projection of an aerial \ac{UE}, as shown in Fig.~\ref{system-model-comp}. The area of each cluster is then given by $A = \pi R_c^2$.}

\emph{\acp{SBS} belonging to the same cluster can cooperate to serve cached content} to the aerial \ac{UE} whose projection on the ground is assumed at the cluster center. 
The aerial \ac{UE} can represent a conventional cellular-connected UAV or a passenger of a flying drone-taxi \cite{8473483}.\footnote{A Chinese major drone maker company called Ehang Corp tested and completed more than 1,000 test drone flights with and without passengers.} 
Due to the strong  \ac{LoS}-dominated interference at high altitudes, we allow multiple SBSs within a certain distance from the aerial \ac{UE} to cooperatively transmit a requested content that they previously cached. 

\subsection{Probabilistic Caching Placement}
Each \ac{SBS} has a surplus memory designated for caching content from a known file library. These files represent the content catalog that an aerial \ac{UE} may request, and are indexed in a descending order of popularity. 
We adopt a random content placement policy in which each content $f$ is cached independently at each SBS according to a probability $c_f$, $0 \leq c_f \leq 1$. Note that \acp{SBS} caching content $f$ can be modeled as a \ac{PPP}  $\Phi_{bf}$ with the intensity function given by the independent thinning theorem as $\lambda_{bf}=c_f \lambda_b$ \cite{haenggi2012stochastic}. Similarly, \acp{SBS} not caching a content $f$ can be modeled as another \ac{PPP}  $\Phi_{bf}^{!}$ with intensity function $\lambda_{bf}^{\circ}=(1-c_f) \lambda_b$, where $\Phi_b = \Phi_{bf} \cup \Phi_{bf}^{!}$.
The \ac{PMF} of the number of \acp{SBS} caching content $f$ in a cluster is given by:
\begin{align}
\label{cond-k}
\Pb(n=\kappa)=\frac{(c_f\lambda_bA)^{\kappa}e^{-c_f\lambda_bA}}{\kappa!},
\end{align}
which represents a Poisson distribution with mean $c_f\lambda_bA$.

\subsection{Serving Distance Distributions}
\black{Under the condition of having $\kappa$ caching \acp{SBS} in the cluster of interest, the distribution of such in-cluster caching \acp{SBS} will follow a \ac{BPP}. This \ac{BPP} consists of $\kappa$ uniformly and independently distributed \acp{SBS} in the cluster.}
%

The set of cooperative \acp{SBS} providing content $f$ is defined as $\Phi_{cf} = \{b_i \in \Phi_{bf} \cap \mathcal{B}(0, R_c)\}$, where $\mathcal{B}(0, R_c)$ denotes the ball centered at the origin with radius $R_c$.
%
Considering the aerial \ac{UE} located at the origin in $\mathbb{R}^2$, i.e., $(0,0,h_d) \in \mathbb{R}^3$,  the 2D distances from the cooperative \acp{SBS} to the aerial \ac{UE} are denoted by $\boldsymbol{R}_{\kappa}= [R_1, \dots, R_{\kappa}]$. Then, conditioning on $\boldsymbol{R}_{\kappa} = \boldsymbol{r}_{\kappa}$, where $\boldsymbol{r}_{\kappa}= [r_1, \dots, r_{\kappa}]$, the conditional \ac{PDF} of the joint serving distances' distribution is denoted as $f_{\boldsymbol{R}_{\kappa}}(\boldsymbol{r}_{\kappa})$. 
The $\kappa$ cooperative \acp{SBS} that cache a content $f$ can be seen as the $\kappa$-closest \acp{SBS} to the cluster center from the PPP $\Phi_{bf}$. Since the $\kappa$ \acp{SBS} are independently and uniformly distributed in the cluster approximated by $\mathcal{B}(0, R_c)$, we have the \ac{PDF} of the horizontal distance $r_i$ from \ac{SBS} $i$ to the aerial \ac{UE} at $(0,0,h_d)$ given as \cite{haenggi2012stochastic}
\[
    f_{R_i}(r_i)=\left\{
                \begin{array}{ll}
                  \frac{2r_i}{R_c^2}, \quad\quad\quad 0\leq r_i\leq R_c,\\ 
                 0, \quad\quad\quad\quad {\rm otherwise},
                \end{array}   
              \right.
  \]
for any $i \in \mathcal{K}_f = \{1, \dots, \kappa\}$, where $\mathcal{K}_f$ is the set of SBSs that cache a content $f$ within the ball $\mathcal{B}(0, R_c)$. \textcolor{black}{From the i.i.d. property of \ac{BPP}, the conditional joint \ac{PDF} of the serving distances $\boldsymbol{R}_{\kappa}= [R_1, \dots, R_{\kappa}]$ is}
\begin{align}
\label{serv-dist}
f_{\boldsymbol{R}_{\kappa}}(\boldsymbol{r}_{\kappa})=\prod_{i=0}^{\kappa} \frac{2r_i}{R_c^2}.
\end{align}

We consider a content delivery from ground \acp{SBS} having the same height $h_{\textrm{SBS}}$ to an aerial \ac{UE} located at altitude $h_d$. The SBS vertical antenna pattern is directional and down-tilted for ground \acp{UE}. The vertical antenna beamwidth and down-tilt angle of the SBSs are denoted respectively by $\theta_B$ and $\theta_t$. The side and main lobe gains of the antennas are denoted by $G_s$ and $G_m$, respectively. Since the horizontal distance between the aerial \ac{UE} and SBS $i$ is $r_i$, the communication link distance will be $d_i=\sqrt{r_i^2 + (h_d-h_{\textrm{SBS}})^2}$ for all $i \in \mathcal{K}_f$.
%
\subsection{Channel Model}
For the CoMP transmission between \acp{SBS} and the aerial \ac{UE}, we consider a wireless channel that is characterized by both large-scale and small-scale fading. For the large-scale fading, the channel between SBS $i$ and the aerial user is described by the \ac{LoS} and \ac{NLoS} components, which are considered separately along with their probabilities of occurrence \cite{ding2016performance}. \black{This assumption is apropos for such  ground-to-air channels that are often dominated by \ac{LoS} communication \cite{azari2017coexistence}.}
Therefore, the antenna gain plus path loss for each component, i.e., \ac{LoS} and \ac{NLoS}, will be		
\begin{align}
\zeta_v(r_i) = A_v G(r_i) d_i^{-\alpha_v}  = A_v G(r_i) \big(r_i^2 + (h_d-h_{\textrm{SBS}})^2\big)^{-\alpha_v/2},
\end{align}
where $v\in \{l,n \}$, $\alpha_{l}$ and $\alpha_{n}$ are the path loss exponents for the  \ac{LoS} and NLoS links, respectively, with $\alpha_{l}<\alpha_{n}$, and $A_{l}$ and  $A_{n}$ are the path-loss constants at the reference distance $d_i = \SI{1}{m}$ for the \ac{LoS} and \ac{NLoS}, respectively. $G(r_i)$ is the total
antenna directivity gain between SBS $i$ and the aerial UE, which can be written similar to \cite{azari2018reshaping} as
\[
    G(r_i) =\left\{
                \begin{array}{ll}
                  G_m, \quad\quad\quad {\rm for \quad} r_i \in \mathcal{S}_{bs}, \\
                  G_s, \quad\quad\quad {\rm for \quad} r_i \notin \mathcal{S}_{bs}, 
                \end{array}   
              \right.
  \]
\black{where $\mathcal{S}_{bs}$ is formed by all the distances $r_i$ satisfying $h_{\textrm{SBS}} - r_i {\rm tan}(\theta_t + \frac{\theta_B}{2}) < h_d < h_{\textrm{SBS}} - r_i {\rm tan}(\theta_t - \frac{\theta_B}{2})$. 
In other words, the horizontal distance between a SBS and an aerial \ac{UE} along with the antenna height, antenna beamwidth and down-tilt angles, and the altitude of this aerial \ac{UE} determine whether it is served by a mainlobe or sidelobe of a \ac{SBS} antenna.}  
  
For the small-scale fading, we adopt a Nakagami-$m$ model utilized in \cite{azari2017coexistence} and \cite{azari2018reshaping} for the channel gain, whose \ac{PDF} is given by:
\begin{align}
f(\omega) = \frac{2\frac{m}{\eta}^{m} \omega^{2m-1}}{\Gamma(m)} {\rm exp}\big(-\frac{m}{\eta}\omega^2\big), 
\end{align}
\black{where $\eta$ is a controlling spread parameter}, and the fading parameter $m$ is assumed to be an integer for analytical tractability. 
%
Since the communication links between an aerial \ac{UE} and SBSs are \ac{LoS}-dominated, e.g., suburban environments with $h_d>\SI{40}{m}$ \cite{lin2018sky}, it is assumed to have $m>1$.
Given that $\omega \sim $ Nakagami$(m,\eta/m)$, it directly follows that the channel gain power
$\gamma=\omega^2 \sim \Gamma(m,\eta/m)$, where $\Gamma(k,\theta)$ is a Gamma \ac{RV} with $k$ and $\theta$  denoting the shape and scale parameters, respectively. Hence, the \ac{PDF} of channel power gain distribution will be:  
\begin{align}
f(\gamma) = \frac{(\frac{m}{\eta})^{m} \gamma^{m-1}}{\Gamma(m)} {\rm exp}\big(-\frac{m}{\eta}\gamma\big). 
\end{align}

3D blockage is characterized by the fraction $a$ of the total land area occupied by buildings, the mean number of buildings $e$  per \SI{}{km}$^2$, and the buildings height modeled by a Rayleigh \ac{PDF} with a scale parameter $c$. Hence, the probability of  \ac{LoS} of a caching SBS at a distance $r_i$ from the aerial \ac{UE} is given by \cite{series2013propagation}:
\begin{align}
\Pb_{l}(r_i) = \prod_{n=0}^{{\rm max}(p-1,0)}\Bigg[1 - {\rm exp}\Big(- \frac{\big(h_{\textrm{SBS}} + \frac{h(n+0.5)}{m+1}\big)^2}{2c^2}\Big) \Bigg],
\end{align}
where $h=h_d - h_{\textrm{SBS}}$ and $p=\floor{\frac{r_i\sqrt{ae}}{1000}}$. Different terrain structures and environments can be considered by by varying the set of $(a,e,c)$.

\black{As discussed previously, the performance of a high-altitude aerial \ac{UE} is limited by the \ac{LoS} interference they encounter. We hence propose a multi-SBSs cooperative transmission scheme aiming at mitigating inter-cell interference and improving the performance of such high-altitude aerial UEs. Under this setting, in the next section  we develop a novel mathematical framework to characterize the performance of cache-assisted CoMP transmission for cellular-connected UAVs. The performance of UAVs is then contrasted to their terrestrial  counterparts. 
%
} 

%
%

\section{Coverage Probability Analysis}
\label{content-inter}
\black{In this section, we characterize the network performance in terms of coverage probability.} 
We assume that the SBSs have the same transmit power $P_t$. Without loss of generality, we consider a typical aerial \ac{UE} located at $(0,0,h_d) \in\mathbb{R}^3$. Conditioning on having $\kappa$ \acp{SBS} serving a content $f$, the received signal at the aerial \ac{UE} will be:	
\begin{align}
\label{rec-pwr}
P &= \underbrace{\sum_{i=1}^{\kappa}  P(r_i) \omega_i w_i X_f}_{\text{desired signal}}+
\underbrace{\sum_{j\in \Phi_{bf}^{!} \cap \mathcal{B}(0, R_c) } P(u_j) \omega_j w_j Y_j}_{I_{\text{in}}} 
\nonumber \\ 
&+ \underbrace{\sum_{k\in\Phi_{b} \setminus \mathcal{B}(0, R_c)} P(u_k) \omega_k w_k Y_k}_{I_{\text{out}}} +   \quad  Z,
\end{align}
where the first term represents the desired signal from $\kappa$ transmitting \acp{SBS} with $P(r_i) =\sqrt{P_t}  \zeta_v(r_i)^{0.5}$, $v \in \{l,n\}$, $\omega_i$ being the Nakagami-$m$ fading variable of the channel from \ac{SBS} $i$ to the aerial UE, $w_i$ is the precoder used by \ac{SBS} $i$, and $X_f$ is the channel input symbol that is sent by the cooperating \acp{SBS}. The second and third terms represent the in-cluster interference $I_{\text{in}}$, and the out-of-cluster interference $I_{\text{out}}$, respectively. 
$Y_j$ is the transmitted symbol from interfering \ac{SBS} $j$ and 
\[
    P(u_j) =\left\{
                \begin{array}{ll}
                  P_l(u_j) = \sqrt{P_t}  \zeta_l(u_j)^{0.5}, \quad\quad\quad {\rm for \quad LoS}, \\
                 P_n(u_j) = \sqrt{P_t}  \zeta_n(u_j)^{0.5}, \quad\quad\quad {\rm for \quad NLoS},
                \end{array}   
              \right.
  \]
where $u_j$ is the horizontal distance between interfering SBS $j$ and the aerial UE. $Z$ is a circular-symmetric zero-mean complex Gaussian \ac{RV} modeling the background thermal noise.  
In-cluster interference occurs only for the case in which not all of the collaborative SBSs (within distance $R_c$) have the cached content (i.e., $c_f <1$). In this case, the set of interfering SBSs will be characterized by $\Phi_{b} \setminus\Phi_{cf} = \Big\{b_i \in \big\{\Phi_{b} \setminus \mathcal{B}(0, R_c)\big\} \cup \big\{\Phi_{bf}^{!} \cap \mathcal{B}(0, R_c)\big\} \Big\}$. For ease of notation, we denote $\big\{\Phi_{bf}^{!} \cap \mathcal{B}(0, R_c)\big\}$ as $\Phi_{cf}^{!}$.
Otherwise, for $c_f =1$, there is no in-cluster interference and the set of interfering SBSs will then be $\Phi_{b} \setminus\Phi_{cf} = \{b_i \in \Phi_{b} \setminus \mathcal{B}(0, R_c)\}$.

We assume that the \ac{CSI} is available at the serving \acp{SBS}, i.e., the precoder $w_i$ can be set as $\frac{\omega_i^*}{|\omega_i|}$, where $\omega_i^*$ is the complex conjugate of $\omega_i$.  
\black{Assuming that $X_f$, $Y_j$, and $Y_k$ in (\ref{rec-pwr}) are independent zero-mean \acp{RV} of unit variance, and averaging over all \ac{LoS} and \ac{NLoS} configurations for the $\kappa$ caching \acp{SBS}, 
the \ac{SIR} at the aerial \ac{UE} will then be:}
\begin{align}
\label{nakagami}
\Upsilon_{|\boldsymbol{r}_{\kappa}} &= \sum_{o=0}^{\kappa}{\kappa \choose o} \prod_{i=0}^{o}\Pb_{l}(r_i)\prod_{j=o+1}^{\kappa}\Pb_{n}(r_j)\cdot
\nonumber \\
&\frac{ P_t \Bigg|\sum_{i=1}^{o} \zeta_{l}^{1/2}(r_i) \omega_i +  \sum_{j=o+1}^{\kappa} \zeta_{n}^{1/2}(r_j) \omega_j \Bigg|^2}{I_{\text{in}} + I_{\text{out}}}.
\end{align}
In (\ref{nakagami}), we have $\big|\sum_{i=1}^{o} \zeta_{l}^{1/2}(r_i) \omega_i +  \sum_{j=o+1}^{\kappa} \zeta_{n}^{1/2}(r_j) \omega_j \big|^2$ representing the square of a weighted sum of $\kappa$ Nakagami-$m$ \acp{RV}. 
Since there is no known closed-form expression for a weighted sum of Nakagami-$m$ \acp{RV}, we use the Cauchy-Schwarz's inequality to get an upper bound on a square of weighted sum as follows:
\begin{align}
\Bigg|\sum_{i=1}^{o} \zeta_{l}^{1/2}(r_i) \omega_i +  &\sum_{j=o+1}^{\kappa} \zeta_{n}^{1/2}(r_j) \omega_j \Bigg|^2 =
 \Bigg(\sum_{i=1}^{\kappa}Q_i\Bigg)^2 
 \nonumber \\
 \label{cauchy}
 &\leq \kappa \Bigg(\sum_{i=1}^{\kappa}Q_i^2\Bigg),
\end{align}
\black{where $Q_i=\zeta_{v}^{1/2}(r_i)\omega_i$, is a scaled Nakagami-$m$ \ac{RV}, with $v\in \{l,n\}$ and $i \in \mathcal{K}_f$}. Since $\omega_i \sim$ Nakagami$(m,\eta/m)$, from the scaling property of the Gamma \ac{PDF}, $Q_i^2 \sim \Gamma\big(k_i=m,\theta_i=2\eta\zeta_{v}(r_i)/m\big)$. To get a statistical equivalent \ac{PDF} of a sum of $\kappa$ Gamma \acp{RV} $Q_i$ with different shape parameters $\theta_i$, we adopt the method of sum of Gammas second-order moment match proposed in \cite[Proposition 8]{heath2011multiuser}. It is shown that the equivalent Gamma distribution, denoted as $J \sim \Gamma(k,\theta)$, with the same first and second-order moments has the parameters
$k = \Big(\sum_i{k_i\theta_i}\Big)^2/\sum_i{k_i\theta_i^2}$	and
$\theta = \sum_i{k_i\theta_i^2}/\sum_i k_i\theta_i$.
To showcase the accuracy of the second-order moment approximation in our case, for an arbitrary realization of the network, we plot in Fig.~\ref{gamma-effect} the \ac{PDF} of the equivalent channel gain. As evident from the plot, approximating a sum of $\kappa$ Gamma \acp{RV} with an equivalent Gamma \ac{RV} whose parameters are 
\begin{align}
\label{equiv-gamma}
k=\frac{m\big(\sum_i{\zeta_{v}(r_i)}\big)^2}{\sum_i{\big(\zeta_{v}(r_i)\big)^2}} 
\quad \quad \text{and} \quad \quad
 \theta=\frac{\eta\sum_i{\zeta_{v}(r_i)}}{m\sum_i  \zeta_{v}(r_i)},
  \end{align}
  is quite accurate. For tractability, we further upper bound the shape parameter $k$ in (\ref{equiv-gamma}):
\begin{align}
\label{approx-k} 
k= m \frac{\Big(\sum_{i} \zeta_v(r_i)\Big)^2}{\sum_{i} \Big(\zeta_v(r_i)\Big)^2}  \leq m\frac{\kappa\sum_{i} \Big(\zeta_v(r_i)\Big)^2}{\sum_{i} \Big(\zeta_v(r_i)\Big)^2} = m\kappa,
\end{align}
where $m\kappa$ is integer.
%
\begin{figure} [!t] 	
\centering
\includegraphics[width=0.38\textwidth]{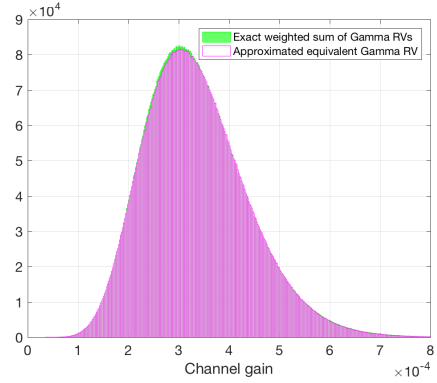}		
\caption {Monte Carlo simulation of the PDF of the equivalent gain of channels between cooperating SBSs and the aerial UE, including path loss and fading. A PPP realization of density $\lambda_b=20$ SBS/\SI{}{km}$^2$ is run for a simulated area of \SI{20}{km}$^2$ with $m=3$ and $R_c=\SI{200}{m}$.}
\label{gamma-effect}
\vspace{-0.5 cm}
\end{figure}

\black{We next derive an upper bound expression on the coverage probability. The novelty of our approach is that it adopts an upper bound on the square of summed Nakagami-$m$ \acp{RV} and second-order match approximation of Gamma \acp{RV}. This allows us to get a closed-form upper bound on the coverage probability, which is difficult to obtain exactly.}

\black{The coverage probability of the aerial \ac{UE} conditioned on the serving distances $\boldsymbol{r}_{\kappa}$ is then expressed as} 
\begin{align}
&\Pb_{\textrm{cov}|\boldsymbol{r}_{\kappa}}= \Pb\big[\Upsilon_{|\boldsymbol{r}_{\kappa}}>\vartheta\big]
\\
&\approx \sum_{o=0}^{\kappa} {\kappa \choose o} \prod_{i=0}^{o}\Pb_{l}(r_i)\prod_{j=o+1}^{\kappa}\Pb_{n}(r_j)  
\Pb\Big(\frac{\kappa P_t \big(\sum_{i=1}^{\kappa}Q_i\big)^2 }{I_{\text{in}} + I_{\text{out}}}>\vartheta\Big),
\nonumber \\
\label{pc-rk}
 &= \sum_{o=0}^{\kappa} {\kappa \choose o} \prod_{i=0}^{o}\Pb_{l}(r_i)\prod_{j=o+1}^{\kappa}\Pb_{n}(r_j) \Pb\Big(\frac{\kappa P_tJ}{I_{\text{in}} + I_{\text{out}}}>\vartheta\Big),
\end{align}
where $\vartheta$ is the $\sir$ threshold. 
\black{The unconditional coverage probability can be obtained as a function of the system parameters, as stated
formally in the following theorem.}
\begin{theorem}
\label{th-one}
The coverage probability is given by:
\begin{align}
\label{cov-prob-theory}		
\Pb_{{\rm cov}} = \sum_{\kappa=1}^{\infty} \Pb(n=\kappa) \int_{\boldsymbol{r_\kappa}=\boldsymbol{0}}^{\boldsymbol{R_c}} \Pb_{{\rm cov}|\boldsymbol{r}_{\kappa}} \prod_{i=0}^{\kappa} \frac{2r_i}{R_c^2} \dd{ \boldsymbol{r}_{\kappa}},
\end{align}
where $\Pb_{{\rm cov}|\boldsymbol{r}_{\kappa}}$ is the conditional coverage probability in (\ref{cond-cov}) (at the top of next page), where $\varpi = \vartheta/\kappa P_t\theta$. 
\end{theorem}
\begin{proof}
We proceed to obtain the coverage probability as follows:
\begin{align}
&\Pb\Big(\frac{\kappa P_tJ}{I_{\text{in}} + I_{\text{out}}}>\vartheta\Big) =  \Pb\Big(\kappa P_tJ>\vartheta 
\big( I_{\text{in}} + I_{\text{out}}\big)\Big)
\nonumber \\
 &= \Eb_{I_{\text{in}},I_{\text{out}}} \Bigg[  \Pb\Big(\kappa P_tJ>\vartheta\big( I_{\text{in}} + I_{\text{out}}\big)\Big) \Bigg]
 \nonumber \\
  &= \Eb_{I_{\text{in}},I_{\text{out}}} \Bigg[  \Pb\Big(J>\frac{\vartheta}{\kappa P_t}\big( I_{\text{in}} + I_{\text{out}}\big)\Big) \Bigg]
 \nonumber \\
 &\overset{(a)}{\approx} \Eb_{I_{\text{in}},I_{\text{out}}} \Big[ \sum_{k=0}^{k-1}  \frac{(\vartheta/\kappa P_t\theta)^k}{k!}\big( I_{\text{in}} + I_{\text{out}}\big)^k
{\rm exp}\Big(-\frac{\vartheta}{\kappa P_t\theta}\big( I_{\text{in}} + I_{\text{out}}\big)\Big) \Big]
   \nonumber \\
  \label{prob-y2}
&\overset{(b)}{=} \Eb_{I_{\text{in}},I_{\text{out}}} \Big[ \sum_{k=0}^{k-1}  \frac{(-\varpi)^k}{k!}
\frac{d^k}{d\vartheta^k} \Lc_{I_{\text{in}}+I_{\text{out}}|\boldsymbol{r}_{\kappa}}(\varpi)  \Big],
\end{align}
where (a) follows from the \ac{PDF} of Gamma \ac{RV} $J$ with parameters $\theta$ and $k$ given in (\ref{equiv-gamma}) and (\ref{approx-k}), respectively. (b) follows from the fact that $\varpi = \vartheta/\kappa P_t\theta$ along with the Laplace transform definition of the \ac{RV} $I_{\text{in}}+I_{\text{out}}$.  Next, we derive the Laplace transform of interference:
\newcounter{MYtempeqncnt}
\begin{figure*}[!t]
\normalsize
\setcounter{MYtempeqncnt}{\value{equation}}
\begin{align}
&\Pb_{\textrm{cov}|\boldsymbol{r}_{\kappa}}=
\sum_{o=0}^{\kappa} {\kappa \choose o} \prod_{i=0}^{o}\Pb_{l}(r_i)\prod_{i=o+1}^{\kappa}\Pb_{n}(r_i)
\sum_{k=0}^{k-1}  \frac{(-\varpi)^k}{k!}\frac{\partial^k}{\partial\varpi^k}  \cdot
\nonumber \\
\label{cond-cov}
&{\rm exp}\Bigg(-2\pi \lambda_{bf}^{\circ}\int_{v=0}^{R_c}\Big(1 -
         \delta_l\Pb_{l}(v) 
        - \delta_n \Pb_{n}(v)
        \Big)v\dd{v}\Bigg)   
{\rm exp}\Bigg(-2\pi \lambda_p\int_{v=R_c}^{\infty}\Big(1 -
        \delta_l\Pb_{l}(v) 
        - \delta_n \Pb_{n}(v)
        \Big)v\dd{v}\Bigg).
\end{align}
\setcounter{equation}{\value{MYtempeqncnt}+1}
\hrulefill
\end{figure*}
\begin{align}
&\Lc_{I_{\text{in}}+I_{\text{out}}|\boldsymbol{r}_{\kappa}}(\varpi)  = \Eb_{I_{\text{in}},I_{\text{out}}} \Big[{\rm exp}\big(-\varpi ( I_{\text{in}} + I_{\text{out}})\big) \Big]
 \nonumber \\
 & = \mathbb{E} \Bigg[
 e^{-\varpi\sum_{j \in \Phi_{cf}^{!}} \gamma_j  P(u_j)^2}
 e^{-\varpi\sum_{j \in \Phi_{b} \setminus \mathcal{B}(0, R_c) } \gamma_j  P(u_j)^2}  
  \Bigg] 
 \nonumber \\ 
  &= \mathbb{E}_{\Phi_b} 
  \Big[  \prod_{j \in \Phi_{cf}^{!}}  \mathbb{E}_{\gamma_j}  e^{-\varpi\gamma_j  P(u_j)^2} 
  \prod_{j \in \Phi_{b} \setminus \mathcal{B}(0, R_c)}  \mathbb{E}_{\gamma_j}  e^{-\varpi \gamma_j  P(u_j)^2}
  \Big] 
  \nonumber \\
  &\overset{(a)}{=} \mathbb{E}_{\Phi_b} \Bigg[
\prod_{j \in \Phi_{cf}^{!}}  \Big[ \Big(1 + \frac{\varpi P_l(u_j)^2 }{m} \Big)^{-m} \Pb_{l}(u_j)  + 
 \nonumber \\
&\Big(1 + \frac{\varpi P_n(u_j)^2 }{m} \Big)^{-m} \Pb_{n}(u_j) \Big]
 \cdot
 \nonumber \\
&\prod_{j \in \Phi_{b} \setminus \mathcal{B}(0, R_c)} 
 \Big[ \Big(1 + \frac{\varpi P_l(u_j)^2 }{m} \Big)^{-m} 
 \nonumber \\
& \cdot\Pb_{l}(u_j) + \Big(1 + \frac{\varpi P_n(u_j)^2 }{m} \Big)^{-m}  \Pb_{n}(u_j)   \Big]
    \Bigg]
 \nonumber
 \\
    \label{LT_c0}  
        &\overset{(b)}{=} {\rm exp}\Bigg(-2\pi \lambda_{bf}^{\circ}\int_{v=0}^{R_c}\Big(1 -
         \delta_l\Pb_{l}(v) 
        - \delta_n \Pb_{n}(v)
        \Big)v\dd{v}\Bigg)  \cdot
    \nonumber \\        
        & {\rm exp}\Bigg(-2\pi \lambda_p\int_{v=R_c}^{\infty}\Big(1 -
        \delta_l\Pb_{l}(v) 
        - \delta_n \Pb_{n}(v)
        \Big)v\dd{v}\Bigg),
\end{align}
%
where (a) follows from the \ac{MGF} of gamma distribution. (b) follows from the probability
generating functional (PGFL) of \ac{PPP}, the substitution $\delta_l=\big(1 + \frac{\varpi P_l(v)^2 }{m} \big)^{-m}$ and $\delta_n=\big(1 + \frac{\varpi P_n(v)^2 }{m} \big)^{-m}$, and Cartesian to polar coordinates conversion. 
By using (\ref{pc-rk}), (\ref{prob-y2}), and (\ref{LT_c0}) along with the \ac{PMF} of the number of caching \acp{SBS} given in (\ref{cond-k}), the coverage probability is obtained. 
\end{proof}
\black{Important insights about the coverage probability can be obtained from (\ref{cov-prob-theory}). 
First, if the collaboration distance $R_c$ or the caching probability $c_f$ increases, both the probability $\Pb(n=\kappa)$ and the integrand value in (\ref{cov-prob-theory}) increase, and, thus,  the coverage probability grows accordingly.  
Furthermore, the effect of the spatial \ac{SBS} density $\lambda_b$ on the coverage probability is two-fold. On the  one hand, the average number of caching \acp{SBS} increases with $\lambda_b$ as characterized by $\Pb(n=\kappa)$, which results in a higher desired signal power. On the other hand, this advantage is counter-balanced by the increase in interference power when $\lambda_b$ increases, as captured in the decaying exponential functions in (\ref{cond-cov}).}

\section{Numerical Results}
 \begin{figure} [!t]	
\centering
\includegraphics[width=0.35\textwidth]{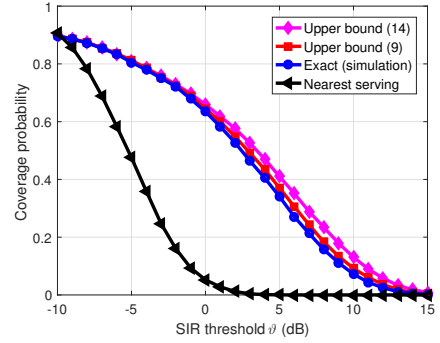}		
\caption {The derived upper bound on the coverage probability is plotted versus the \ac{SIR} threshold $\vartheta$.}
\label{cov-prob}
\vspace{-0.35cm}
\end{figure}
\begin{table}[ht]
\caption{Simulation Parameters} 
\centering 
\begin{tabular}{c c  c} 
\hline\hline 
Description & Parameter & Value  \\ [0.5ex] 
\hline 
LoS path-loss exponent& $\alpha_{l}$& 2.09 \\ 
NLoS path-loss exponent& $\alpha_{n}$& 3.75 \\ 
LoS path-loss constant& $A_{l}$& \SI{-41.1}{dB} \\ 
NLoS path-loss constant& $A_{n}$ & \SI{-32.9}{dB} \\ 
Antenna main lobe gain&$G_m$&\SI{10}{dB}\\		
Antenna side lobe gain&$G_s$&\SI{-3.01}{dB}\\
Nakagami fading parameter&$m$&3\\
Nakagami spreading factor&$\eta$&2\\
SBS antenna height&$h_{\textrm{SBS}}$ &\SI{30}{m}\\
Aerial \ac{UE} altitude&$h_d$&\SI{100}{m}\\
Area fraction occupied by buildings &$a$& 0.3\\
Mean number of buildings &$e$&200 per \SI{}{km}$^2$\\
Buildings height Rayleigh parameter &$c$&15\\
Collaboration distance&$R_c$&\SI{200}{m}\\
Density of \acp{SBS}&$\lambda_{b}$&20 \acp{SBS}/\SI{}{km}$^2$ \\
$\sir$ threshold&$\vartheta$&\SI{0}{\deci\bel}\\
Down-tilt angle &$\theta_t$& $8^\circ$\\
Vertical beamwidth &$\theta_B$& $30^\circ$\\
Content caching probability &$c_f$& 1\\
\hline 
\end{tabular}
\label{ch4:table:sim-parameter} 
\vspace{-0.5cm}
\end{table}

For our simulations, we consider a network having the  parameter values indicated in Table \ref{ch4:table:sim-parameter}. Monte Carlo simulations are used to validate the developed mathematical model.

In Fig.~\ref{cov-prob}, we show the theoretical upper bound on the coverage probability obtained in (\ref{cov-prob-theory}), simulation of the exact coverage probability, and simulation of the upper bound based on Cauchy's inequality in (\ref{cauchy}). The figure shows that the  Cauchy's inequality-based upper bound is remarkably tight. Moreover, although the upper bound on the coverage probability obtained in (\ref{cov-prob-theory}) is less tight, it still represents a reasonably tractable bound on the exact coverage probability. Hence, (\ref{cov-prob-theory}) can be treated as a proxy of the exact result. Recall that (\ref{cauchy}) is based on an upper bound on a square of a sum of Nakagami-$m$ \acp{RV} while the expression in (\ref{cov-prob-theory}) goes further by two more steps. First, we approximate the sum of Gamma \acp{RV} to an equivalent Gamma \ac{RV}. Then, we approximatie the shape parameter of the yielded Gamma \ac{RV} to an integer (\ref{approx-k}).
Fig.~\ref{cov-prob} also compares the coverage probability of the proposed \ac{CoMP} transmission scheme with the nearest serving \ac{SBS} transmission scheme. As evident from the plot, allowing multiple transmission of the same content from neighboring \acp{SBS} significantly enhances the coverage probability, e.g., from $10\%$ to $70\%$ at $\vartheta=$ \SI{0}{dB} \black{for an average of only 2.5 serving \acp{SBS}}.

 \begin{figure} [!tp]	
\centering
\includegraphics[width=0.32\textwidth]{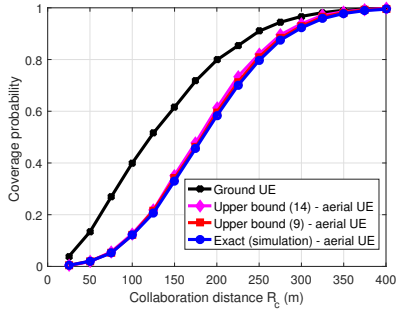}		
\caption {Coverage probability versus the collaboration distance $R_c$ for the aerial and ground UEs.}
\label{collab-dist}
\vspace{-0.4cm}
\end{figure}
Fig.~\ref{collab-dist} shows the prominent effect of the collaboration distance $R_c$ on the coverage probability for two UEs, namely, ground UE, and high-altitude aerial UE. We can see that for both UEs, the coverage probability monotonically increases with $R_c$ since more \acp{SBS} cooperate to send a content when $R_c$ increases. Although the achievable coverage probability of the aerial \ac{UE} is always upper bounded by that of the ground \ac{UE} (due to the down-tilt of the SBSs' antennas), we can see that the rate of improvement of the coverage probability with $R_c$, i.e., the slope, is higher for the aerial UE. This can be interpreted by the fact that increasing the collaboration distance yields more \ac{LoS} signals within the desired signal side and mitigates them from the interference. In contrast, for the ground UE, the transmission is always dominated by NLoS signals and Rayleigh fading. 

To show the effect of content availability, i.e., content caching, in Fig.~\ref{optimal-h}, we plot the coverage probability versus the \ac{SIR} threshold $\vartheta$ for different $c_f$. 
We observe that the coverage probability decreases as the caching probability $c_f$ decreases. This stems from the fact that the average number of caching \acp{SBS} decreases as $c_f$ decreases. This in turn reduces the cooperative transmission gain. Note that the value $c_f$ is, in fact, a parameter that can be designed based on various factors such as the memory size of SBSs, the popularity of files, and file library size.

\section{Conclusion} 
In this paper, we have proposed a novel framework for cooperative transmission and probabilistic caching suitable for high-altitude aerial UEs. In order to obtain analytically tractable expressions, we have employed Cauchy's  inequality and a second-order moment approximation of Gamma \acp{RV} to derive a closed-form upper bound on the content coverage probability. We have then shown that the derived bound is considerably tight. 
\black{We have also shown that  exploiting \ac{CoMP} transmission with an average of 2.5 serving \acp{SBS} per cluster  
significantly improves the coverage probability, e.g., from $10\%$ to $70\%$ at \SI{0}{dB} \ac{SIR} threshold}. Moreover, comparing the performance of an aerial \ac{UE} to a ground UE, our results have shown that the coverage probability of an aerial \ac{UE} is always upper bounded by that of a ground \ac{UE} owing to the down-tilted antenna pattern.

 \begin{figure} [!t]	
\centering
\includegraphics[width=0.35\textwidth]{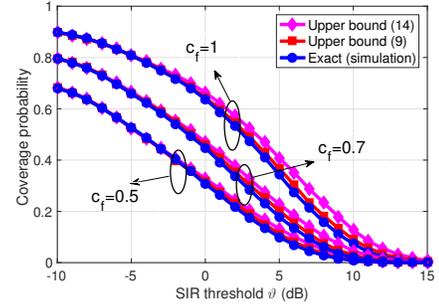}		
\caption {Coverage probability versus \ac{SIR} threshold $\vartheta$ for different content caching probability $c_f$.}
\label{optimal-h}
\vspace{-0.4cm}
\end{figure}

\bibliographystyle{IEEEtran}
\bibliography{bibliography}
\end{document}